\documentclass[
  twocolumn,
  aps,
  pra,
  floatfix,
  amsmath,amssymb,
  longbibliography
]{revtex4-2}

\usepackage[T1]{fontenc}
\usepackage[utf8]{inputenc}
\usepackage{amsthm}
\usepackage{graphicx}
\usepackage{booktabs}
\usepackage[dvipsnames]{xcolor}
\usepackage{array}
\usepackage{tikz}
\usetikzlibrary{calc,positioning}
\usepackage[
  breaklinks=true,
  colorlinks=true,
  citecolor=blue,
  linkcolor=blue,
  urlcolor=blue
]{hyperref}
\usepackage{orcidlink}
\hypersetup{
  pdftitle={Pseudocontexts forced by finite context hypergraphs},
  pdfauthor={Mirko Navara and Karl Svozil},
  pdfkeywords={pseudocontext, context hypergraph, incidence certificate,
    probability, faithful orthogonal representation}
}

\newtheorem{theorem}{Theorem}

\newtheorem{proposition}[theorem]{Proposition}
\newtheorem{corollary}[theorem]{Corollary}

\theoremstyle{definition}
\newtheorem{definition}[theorem]{Definition}

\theoremstyle{remark}

\newcommand{\R}{\mathbb R}
\newcommand{\C}{\mathbb C}
\newcommand{\E}{\mathcal E}
\newcommand{\one}{\mathbf 1}
\newcommand{\Pc}{\mathbb P}

\definecolor{ctx1}{RGB}{230,25,75}
\definecolor{ctx2}{RGB}{60,180,75}
\definecolor{ctx3}{RGB}{0,130,200}
\definecolor{ctx4}{RGB}{245,130,48}
\definecolor{ctx5}{RGB}{145,30,180}
\definecolor{ctx6}{RGB}{70,240,240}
\definecolor{ctx7}{RGB}{240,50,230}
\definecolor{ctx8}{RGB}{210,245,60}
\definecolor{ctx9}{RGB}{250,190,190}
\definecolor{ctx10}{RGB}{0,128,128}
\definecolor{ctx11}{RGB}{128,128,0}
\definecolor{ctx12}{RGB}{170,110,40}
\tikzset{ctxline/.style={line width=0.95pt,line cap=round}}

\begin{document}

\title{Pseudocontexts forced by finite context hypergraphs}

\author{Mirko Navara\,\orcidlink{0000-0002-0880-5992}}
\affiliation{Faculty of Electrical Engineering, Czech Technical University in Prague,
Technick\'{a} 2, CZ-166~27 Prague 6, Czech Republic}
\email{navara@fel.cvut.cz}

\author{Karl Svozil\,\orcidlink{0000-0001-6554-2802}}
\affiliation{Institute for Theoretical Physics, TU Wien,
Wiedner Hauptstrasse 8-10/136, A-1040 Vienna, Austria}
\email{karl.svozil@tuwien.ac.at}

\date{\today}

\begin{abstract}
A context is a complete set of mutually exclusive outcomes whose probabilities sum to one. We define a pseudocontext as two disjoint groups, with no mutually exclusive pair within either group, whose total probabilities must nevertheless be equal solely because of how the contexts overlap. We give an exact finite test for this property and show that the equality is independent of the chosen coordinates and probability model. Applying the test in three dimensions, we obtain a 15-outcome real example with groups of three and a 20-outcome complex example with groups of two. Under this definition, the sharp minimum number of outcomes in each target group is two over the complex field and three over the real field.
\end{abstract}

\maketitle

\section{Introduction}

A context is a complete collection of mutually exclusive propositions.  If
$p$ is a probability assignment and $E$ is a context, completeness and finite
additivity give
\begin{equation}
  \sum_{v\in E}p(v)=1.
  \label{eq:context-normalization}
\end{equation}
This elementary equation is shared by classical sample-space models,
quantum probabilities on orthogonal bases, and any other normalized model of
the same propositions.  It is therefore natural to ask whether finitely many
equations of the form~\eqref{eq:context-normalization} can force an additive
identity between two families that are not themselves contexts.

Such families were called \emph{pseudocontexts} in
Ref.~\cite{2023-navara-svozil}.  There is, however, a necessary distinction.
For particular rays $a_i,b_i$ in a Hilbert space, one may happen to have
\begin{equation}
  \sum_i |a_i\rangle\!\langle a_i|
  =\sum_i |b_i\rangle\!\langle b_i|.
  \label{eq:coordinate-identity}
\end{equation}
Equation~\eqref{eq:coordinate-identity} implies equality of Born-probability
sums for those coordinates, but the equality may disappear when the same
orthogonality hypergraph is coordinatized differently.  It also says nothing
about probability models that are not derived from those projectors.

Here we reserve \emph{pseudocontext} for the stronger, finite-combinatorial
notion: the equality must be a balanced linear consequence of the context
equations themselves.  This implements two requirements simultaneously:
\begin{enumerate}
  \item the identity holds for every faithful orthogonal representation of the
        hypergraph, not merely for one selected vector realization;
  \item the identity holds for every admissible probability assignment,
        independently of whether it is classical, quantum, softmax-generated~\cite{svozil-2026-softmax},
        or obtained by another rule.
\end{enumerate}
Only completeness, additivity on mutually exclusive propositions, and a
single context-independent weight for each vertex are used.

The central result is an incidence-matrix characterization.  Besides making
the definition testable, it cleanly separates three issues that can otherwise
be conflated: combinatorial forcing, geometric realizability, and the state
space of a particular probability theory.  We then revisit the
three-dimensional constructions of Ref.~\cite{2023-navara-svozil} and a
20-vertex octagon whose scalar-field representability is analyzed
separately in Ref.~\cite{2026-Navara-Svozil-complex-only-3d-hypergraphs}.

\section{Contexts, weights, and orthogonal representations}
\label{sec:framework}

Let $\Gamma=(V,\E)$ be a finite $n$-uniform hypergraph.  Its vertices are
elementary propositions and every $E\in\E$ is declared to be a complete
context of $n$ mutually exclusive propositions.  We exclude repeated
vertices inside a context.

\begin{definition}[Admissible weight]
An \emph{admissible weight} on $\Gamma$ is a function
$p:V\to[0,1]$ satisfying Eq.~\eqref{eq:context-normalization} for every
$E\in\E$.  The set of admissible weights is denoted by $\mathcal S(\Gamma)$.
The hypergraph is \emph{probabilistically consistent} if
$\mathcal S(\Gamma)\ne\varnothing$.
\end{definition}

The range $[0,1]$ records the probabilistic interpretation, but the forcing
criterion below is deliberately stronger: it requires an identity on the
entire real affine solution space of the normalization equations.  Positivity
will therefore play no part in the proof.

Different theories select different subsets of $\mathcal S(\Gamma)$.
An admissible weight whose range is contained in \ensuremath{\{0,1\}}
is called a two-valued state.  Classical deterministic models use
two-valued states and their convex combinations.  A quantum realization assigns rays to the vertices and uses
the Born rule. A softmax, neural, or statistical parametrization is also covered whenever
its output is a single, globally defined probability assignment
$p:V\to[0,1]$ satisfying every context-normalization equation; that is,
\[
 \sum_{v\in E}p(v)=1
 \qquad\text{for every }E\in\E.
\]
Any probability assignment that is not a single globally defined function
on $V$---so that a vertex shared by different contexts may receive
different values in those contexts---is not a weight on $\Gamma$.  It
instead describes context-labelled outcomes with context-dependent
frequencies or expectation values.

To distinguish properties of the abstract hypergraph from properties of a
particular vector or projector realization, let $\Pc(H)$ denote the set of
rays, that is, the one-dimensional subspaces, of an $n$-dimensional real or
complex Hilbert space $H$.  A context $E\in\E$ is an $n$-element subset of
vertices; it is not itself a projection.  We use the faithful version of the
standard graph-theoretic notion of an orthogonal representation
~\cite{lovasz-79,lovasz-89,Portillo-2015}.

\begin{definition}[Faithful orthogonal representation]
A~\emph{faithful orthogonal representation} (FOR), also called a
\emph{faithful coordinatization}, of $\Gamma=(V,\E)$ in $H$ is an injective
map
\[
 h:V\to\Pc(H).
\]
For each vertex $v\in V$, choose a unit vector $x_v\in H$ spanning the ray
$h(v)$.  The representation is called faithful if, for all distinct
$u,v\in V$,
\begin{equation}
 \begin{split}
  \langle x_u,x_v\rangle=0
  &\quad\Longleftrightarrow\quad h(u)\perp h(v)\\
  &\quad\Longleftrightarrow\quad
  \text{there is an }E\in\E\\[-2pt]
  &\hspace{2.7cm}\text{such that }\{u,v\}\subseteq E.
 \end{split}
 \label{eq:faithful}
\end{equation}
The condition is independent of the chosen unit representatives, since
changing a representative by a sign or a complex phase does not affect
whether the inner product is zero.

Since $|E|=n=\dim H$, the vectors
$\{x_v:v\in E\}$ form an orthonormal basis of $H$ for every $E\in\E$.
Equivalently, define
\[
 P_{h(v)}=|x_v\rangle\langle x_v|.
\]
This is the rank-one orthogonal projection onto the ray $h(v)$; in
particular, it is self-adjoint and idempotent.  For every context
$E\in\E$, the corresponding projections are pairwise orthogonal and satisfy
\begin{equation}
 \sum_{v\in E}P_{h(v)}=I_H.
 \label{eq:projector-context}
\end{equation}
Here \ensuremath{I_H} denotes the identity operator on \ensuremath{H}.
The hypergraph is \emph{geometrically admissible in $H$} if such a faithful
orthogonal representation exists.
\end{definition}

\section{Restrictive definition and certificate theorem}
\label{sec:definition}

Index the rows by $\E$ and the columns by $V$, and let
\begin{equation}
 \ensuremath{M_\Gamma}
 \in\{0,1\}^{\E\times V},\qquad
 \ensuremath{(M_\Gamma)_{E,v}}=
 \begin{cases}
  1,&v\in E,\\
  0,&v\notin E,
 \end{cases}
 \label{eq:incidence}
\end{equation}
be the incidence matrix of $\Gamma$.  Thus, for
$p\in\mathbb R^V$ and $E\in\E$,
\[
 (\ensuremath{M_\Gamma}p)_E
 =\sum_{v\in V}\ensuremath{(M_\Gamma)_{E,v}p(v)}
 =\sum_{v\in E}p(v).
\]
Consequently, the system
\begin{equation}
  \ensuremath{M_\Gamma}p=\one_\E
  \label{eq:matrix-normalization}
\end{equation}
means exactly that
\[
 \sum_{v\in E}p(v)=1
 \qquad\text{for every }E\in\E.
\]
For $A\subseteq V$, let $\chi_A\in\{0,1\}^{V}$ denote its incidence vector.
\begin{definition}[Balanced context certificate]
Let $A,B\subseteq V$ be disjoint.  A vector
$\lambda\in\mathbb Q^{\E}$ is a \emph{balanced context certificate} for
$(A,B)$ if
\begin{equation}
 \ensuremath{M_\Gamma}^{\mathsf T}\lambda=\chi_A-\chi_B,
 \qquad
 \one_\E^{\mathsf T}\lambda=0.
 \label{eq:certificate}
\end{equation}
The first equation means that, in the linear combination of context
equations with coefficients $\lambda_E$, every vertex outside
$A\cup B$ cancels, while vertices in $A$ and $B$ occur with coefficients
$+1$ and $-1$, respectively.  The second equation ensures that the
constant right-hand sides of the context equations cancel as well.
Therefore, for every solution of $\ensuremath{M_\Gamma}p=\one_\E$,
\begin{align*}
 \sum_{a\in A}p(a)-\sum_{b\in B}p(b)
 &= (\chi_A-\chi_B)^{\mathsf T}p
  = \lambda^{\mathsf T}\ensuremath{M_\Gamma}p,\\
 &= \lambda^{\mathsf T}\one_\E
  = 0.
\end{align*}
Hence a balanced context certificate forces
\[
 \sum_{a\in A}p(a)=\sum_{b\in B}p(b).
\]
\end{definition}

\paragraph{Example (balanced context: firefly logic).}
The following elementary five-vertex configuration is the standard firefly
logic~\cite{cohen}.  Let
\[
 V=\{a,b,c,d,x\},
 \qquad
 E_1=\{a,b,x\},
 \qquad
 E_2=\{c,d,x\}.
\]
A real three-dimensional coordinatization is given by
\(a=(1,0,0)\),  \(b=(0,1,0)\) ,\(x=(0,0,1)\),
and
\(c=(1/\sqrt{2})(1,1,0)\),
\(d=(1/\sqrt{2})(1,-1,0)
\).
Thus $E_1$ and $E_2$ are both orthonormal bases, and the corresponding
normalization equations are
\[
 p(a)+p(b)+p(x)=1,
 \qquad
 p(c)+p(d)+p(x)=1.
\]
Subtracting the second equation from the first gives
\[
 p(a)+p(b)=p(c)+p(d).
\]
This equality is expressed by the balanced context certificate
\[
 A=\{a,b\},
 \qquad
 B=\{c,d\},
 \qquad
 \lambda=(1,-1)^{\mathsf T}.
\]
Indeed, ordering the vertices as $(a,b,c,d,x)$, the incidence matrix is
\[
 \ensuremath{M_\Gamma}=
 \begin{pmatrix}
  1&1&0&0&1\\
  0&0&1&1&1
 \end{pmatrix},
\]
and hence
\[
 \begin{split}
  \ensuremath{M_\Gamma}^{\mathsf T}\lambda
  &= (1,1,-1,-1,0)^{\mathsf T}
   =\chi_A-\chi_B,\\
  \one_\E^{\mathsf T}\lambda&=1-1=0.
 \end{split}
\]
The common vertex $x$ cancels from the two sides, while the constants on
the right-hand sides cancel because the two context equations are taken
with opposite coefficients.  This is a balanced context in the present
terminology, but not a pseudocontext under Definition~\ref{def:pseudocontext}:
the two elements $a,b\in A$ occur together in the context $E_1$, and the two
elements $c,d\in B$ occur together in the context $E_2$.  Thus Condition~2
is violated on both sides.

\begin{definition}[Pseudocontext]
\label{def:pseudocontext}
Let $k\ge2$.  A \emph{pseudocontext of size $k$} is a triple
$(\Gamma;A,B)$ such that
\begin{enumerate}
  \item $\Gamma$ is a finite, probabilistically consistent, uniform context
        hypergraph;
  \item $A,B\subseteq V$ are disjoint $k$-element sets such that no context
        contains two distinct elements of $A$ and no context contains two
        distinct elements of $B$;
  \item $(A,B)$ has a balanced context certificate.
\end{enumerate}
The restriction $k\ge2$ excludes the elementary size-one identities
obtained by subtracting context-normalization equations.  Such identities
remain balanced context certificates, but are not called pseudocontexts
here.
It is a \emph{pseudocontext in $H$} if, in addition, $\Gamma$ admits a FOR
in $H$.
\end{definition}

Condition 2 removes ordinary contexts and their incomplete mutually
exclusive subsets.  Under a FOR, each target family is
pairwise nonorthogonal.  Disjointness prevents identities obtained by padding
both sides with the same propositions.

\begin{theorem}[Finite forcing criterion]
\label{thm:forcing}
Suppose the real affine system $\ensuremath{M_\Gamma}p=\one_\E$ is nonempty and let
$d=\chi_A-\chi_B$.  The following are equivalent:
\begin{enumerate}
  \item $d^{\mathsf T}p=0$ is a linear consequence of the context equations,
        that is, it holds for every real solution of
        $\ensuremath{M_\Gamma}p=\one_\E$;
  \item there is a real vector $\lambda$ such that
        $\ensuremath{M_\Gamma}^{\mathsf T}\lambda=d$ and
        $\one_\E^{\mathsf T}\lambda=0$;
  \item there is a rational balanced context certificate for $(A,B)$.
\end{enumerate}
Whenever these conditions hold,
\begin{equation}
  \sum_{a\in A}p(a)=\sum_{b\in B}p(b)
  \label{eq:universal-equality}
\end{equation}
for every admissible weight \ensuremath{p}
on $\Gamma$.
\end{theorem}

\begin{proof}
If 2 holds and $\ensuremath{M_\Gamma}p=\one_\E$, then
\[
 d^{\mathsf T}p
 =\lambda^{\mathsf T}\ensuremath{M_\Gamma}p
 =\lambda^{\mathsf T}\one_\E
 =0,
\]
so 2 implies 1 and Eq.~\eqref{eq:universal-equality} follows.

Conversely, choose one solution $p_0$ of $\ensuremath{M_\Gamma}p=\one_\E$.  Every vector
$x\in\ker \ensuremath{M_\Gamma}$ gives another solution $p_0+x$.  By 1,
$d^{\mathsf T}(p_0+x)=d^{\mathsf T}p_0=0$, hence $d$ annihilates
$\ker \ensuremath{M_\Gamma}$.  Finite-dimensional linear algebra gives
$d\in(\ker \ensuremath{M_\Gamma})^\perp=\operatorname{im}\ensuremath{M_\Gamma}^{\mathsf T}$, so
$d=\ensuremath{M_\Gamma}^{\mathsf T}\lambda$ for some real $\lambda$.  Evaluating at $p_0$
yields
$0=d^{\mathsf T}p_0=\lambda^{\mathsf T}\one_\E$, proving 2.

The equations in 2 have integer coefficients.  A consistent linear system
over $\mathbb Q$ has a rational solution, for example by Gaussian
elimination.  Thus 2 and 3 are equivalent.
\end{proof}

The word ``forced'' in Definition~\ref{def:pseudocontext} is therefore literal: the identity is
valid even before nonnegativity or a particular state model is imposed.  If
one only asks for equality on the polytope $\mathcal S(\Gamma)$, additional
equalities may arise because positivity forces certain vertices to zero.
Those boundary effects are intentionally excluded from the definition.

\begin{corollary}[Integral form]
Every pseudocontext admits an integer $q\ge1$ and
$\mu\in\mathbb Z^{\E}$ such that
\begin{equation}
 \ensuremath{M_\Gamma}^{\mathsf T}\mu=q(\chi_A-\chi_B),
 \qquad \one_\E^{\mathsf T}\mu=0.
 \label{eq:integer-certificate}
\end{equation}
\end{corollary}

\begin{proof}
Clear the denominators of a rational certificate.  Applying an admissible
real-valued weight to Eq.~\eqref{eq:integer-certificate} gives $q$ times the
desired equality, and division by $q$ completes the argument.
\end{proof}

\begin{corollary}[Independence of coordinatization]
\label{cor:coordinates}
Let $(\Gamma;A,B)$ be a pseudocontext and let $h$ be any FOR of $\Gamma$
in an $n$-dimensional real or complex Hilbert
space.  Then
\begin{equation}
 \sum_{a\in A}P_{h(a)}=\sum_{b\in B}P_{h(b)}.
 \label{eq:operator-forced}
\end{equation}
Consequently the Born-probability sums are equal for every density operator,
independently of the chosen FOR.
\end{corollary}

\begin{proof}
Multiply each operator equation~\eqref{eq:projector-context} by the
corresponding certificate coefficient and sum over $E\in\E$.  The left-hand
side becomes
$\sum_v(\ensuremath{M_\Gamma}^{\mathsf T}\lambda)_vP_{h(v)}$, while the right-hand side is
$(\ensuremath{\one_\E}^{\mathsf T}\lambda)I_H=0$.  Equation~\eqref{eq:certificate} gives
Eq.~\eqref{eq:operator-forced}.
\end{proof}

This proof does not invoke Gleason's theorem.  Gleason's theorem
\cite{Gleason} characterizes quantum states in dimensions at least three,
but the pseudocontext equality already follows from the finite context
equations.

\section{Covering certificates and exact discovery}
\label{sec:coverings}

The most transparent certificates compare two multisets of contexts.
Suppose $\mathcal F_+$ and $\mathcal F_-$ contain the same number of
contexts, counted with multiplicity.  If their vertex multiplicities agree
away from $A\cup B$ and their residual incidence difference is
$\chi_A-\chi_B$, then assigning $+1$ to $\mathcal F_+$, $-1$ to
$\mathcal F_-$, and zero elsewhere gives Eq.~\eqref{eq:certificate}.  We call
this a \emph{complementary covering certificate}.

The general criterion is no harder to test.  For fixed $A,B$, exact rational
row reduction of
\begin{equation}
 \begin{pmatrix}\ensuremath{M_\Gamma}^{\mathsf T}\\ \one_\E^{\mathsf T}\end{pmatrix}\lambda
 =
 \begin{pmatrix}\chi_A-\chi_B\\0\end{pmatrix}
 \label{eq:algorithm}
\end{equation}
either returns a certificate or proves that none exists.  To discover
pseudocontexts rather than verify a proposed pair, compute the consequence
space
\begin{equation}
 \mathcal C_\Gamma
 =\left\{\ensuremath{M_\Gamma^{\mathsf T}\lambda}:
          \lambda\in\mathbb Q^{\E},\
          \ensuremath{\one_\E^{\mathsf T}\lambda=0}\right\}
 \label{eq:consequence-space}
\end{equation}
Here \ensuremath{\lambda} assigns a rational coefficient to every
context-normalization equation.  The condition
\ensuremath{\one_\E^{\mathsf T}\lambda=0}
means that these coefficients sum to zero, so the constant right-hand
sides cancel.  Multiplication by \ensuremath{M_\Gamma^{\mathsf T}}
records the resulting coefficient of each vertex.  Thus
\ensuremath{\mathcal C_\Gamma}
is precisely the space of vertex-coefficient vectors obtained from
balanced rational linear combinations of the context equations.
This space can then be searched for sparse vectors with entries in $\{-1,0,1\}$.  The positive
and negative supports of such a vector are candidate target families; the
nonexclusivity and disjointness conditions are then checked combinatorially.
Integer programming can be used to minimize the number of contexts or
vertices in a certificate.

For an $n$-uniform hypergraph, every $d\in\mathcal C_\Gamma$ satisfies
$\one_V^{\mathsf T}d=0$ because
\[
 \one_V^{\mathsf T}\ensuremath{M_\Gamma}^{\mathsf T}\lambda
 =n\one_\E^{\mathsf T}\lambda=0.
\]
Thus a $\{-1,0,1\}$ consequence automatically has equally large positive
and negative supports.  Equal pseudocontext size is not an extra numerical
coincidence; it follows from uniformity and balance.

\section{A size-three real pseudocontext}
\label{sec:mayet}

Consider the $15$-vertex, eight-context hypergraph introduced in this role in
Ref.~\cite{2023-navara-svozil}.  Its contexts are
\begin{align}
 E_1&=\{1,2,3\},& E_2&=\{3,4,5\},\nonumber\\
 E_3&=\{6,7,8\},& E_4&=\{8,9,10\},\nonumber\\
 E_5&=\{11,12,13\},& E_6&=\{13,14,15\},\label{eq:mayet-edges}\\
 E_7&=\{2,7,12\},& E_8&=\{4,9,14\}.\nonumber
\end{align}
Let
\begin{equation}
 A=\{1,6,11\},\qquad B=\{5,10,15\}.
 \label{eq:mayet-targets}
\end{equation}
No two vertices within either triple share a context.

The four contexts
\begin{equation}
 \mathcal F_+=\{E_1,E_3,E_5,E_8\}
 \label{eq:cover-plus}
\end{equation}
cover every vertex once except $B$, whereas
\begin{equation}
 \mathcal F_-=\{E_2,E_4,E_6,E_7\}
 \label{eq:cover-minus}
\end{equation}
cover every vertex once except $A$.  Therefore
\begin{equation}
 \sum_{j\in\{1,3,5,8\}}\chi_{E_j}
 -\sum_{j\in\{2,4,6,7\}}\chi_{E_j}
 =\chi_A-\chi_B.
 \label{eq:mayet-certificate}
\end{equation}
The two sides use four contexts each, so the certificate is balanced.  Every
admissible weight consequently satisfies
\begin{equation}
 p(1)+p(6)+p(11)=p(5)+p(10)+p(15).
 \label{eq:mayet-equality}
\end{equation}

This hypergraph has classical partition representations and real
three-dimensional FORs~\cite{2023-navara-svozil}.  Hence
Eq.~\eqref{eq:mayet-equality} holds in its classical models, in every one of
its quantum FORs, and for any other admissible weight.
The numerical range of the common sum may still depend on the chosen theory;
the equality itself does not.

\section{A size-two complex pseudocontext}
\label{sec:octagon}

A smaller target family is forced by a $20$-vertex, $12$-context hypergraph.  Let
$v_0,\ldots,v_7$ be the midpoints of a cyclic octagon and
$v_{01},v_{12},\ldots,v_{70}$ its corners, with indices understood modulo eight.
The
eight perimeter contexts are
\begin{equation}
 E_i=\{v_{i-1,i},v_i,v_{i,i+1}\},\qquad i\in\mathbb Z_8,
 \label{eq:octagon-perimeter}
\end{equation}
Four further contexts are
\begin{align}
 D_0&=\{v_0,a_1,v_4\},&D_2&=\{v_2,a_2,v_6\},\nonumber\\
 D_1&=\{v_1,b_2,v_5\},&D_3&=\{v_3,b_1,v_7\}.
 \label{eq:octagon-inner}
\end{align}

Take $A=\{a_1,a_2\}$ and $B=\{b_1,b_2\}$.  The six contexts
\begin{align}
 \mathcal G_+&=\{D_0,D_2,E_1,E_3,E_5,E_7\},\nonumber\\
 \mathcal G_-&=\{D_1,D_3,E_0,E_2,E_4,E_6\}
 \label{eq:octagon-covers}
\end{align}
each cover all $16$ outer vertices once.  The first also covers $a_1,a_2$,
while the second covers $b_1,b_2$.  The two coverings are shown in
Fig.~\ref{fig:octagon}.

\begin{figure*}[t]
\centering

\begin{tikzpicture}[scale=1.6,
  every label/.style={font=\footnotesize},
  context/.style={thick, gray!40, line cap=round},
  cover/.style={ultra thick, black, line cap=round},
  covered/.style={circle, draw=black, fill=black, minimum size=5pt, inner sep=0pt},
  covered_inner/.style={circle, draw=black, fill=black, minimum size=6pt, inner sep=0pt},
  uncovered_inner/.style={circle, draw=black, fill=white, minimum size=6pt, inner sep=0pt},
  targetA/.style={line width=3.8pt,blue,line cap=round},
  targetB/.style={line width=3.8pt,red,dash pattern=on 8pt off 10pt,line cap=round}]

  \node[font=\normalsize] at (-2.3, 2.3) {(a) $\mathcal G_+$};

  \def\Rmid{2}
  \def\Rcor{2.1648}
  \def\Rinn{0.95}

  \foreach \i/\j/\ang in {0/1/22.5, 1/2/67.5, 2/3/112.5, 3/4/157.5, 4/5/202.5, 5/6/247.5, 6/7/292.5, 7/0/337.5} {
    \coordinate (v\i\j) at (\ang:\Rcor);
  }
  \foreach \i/\ang in {0/0, 1/45, 2/90, 3/135, 4/180, 5/225, 6/270, 7/315} {
    \coordinate (v\i) at (\ang:\Rmid);
  }
  \coordinate (a1) at (180:\Rinn);
  \coordinate (a2) at (90:\Rinn);
  \coordinate (b1) at (135:\Rinn);
  \coordinate (b2) at (45:\Rinn);

  \draw[cover] (v0) -- (v4);
  \draw[cover] (v2) -- (v6);
  \draw[cover] (v01) -- (v12);
  \draw[cover] (v23) -- (v34);
  \draw[cover] (v45) -- (v56);
  \draw[cover] (v67) -- (v70);

  \draw[context] (v1) -- (v5);
  \draw[context] (v3) -- (v7);
  \draw[context] (v70) -- (v01);
  \draw[context] (v12) -- (v23);
  \draw[context] (v34) -- (v45);
  \draw[context] (v56) -- (v67);

  \draw[targetA] (a1) -- (a2);
  \draw[targetB] (b1) -- (b2);

  \foreach \i/\j/\ang in {0/1/22.5, 1/2/67.5, 2/3/112.5, 3/4/157.5, 4/5/202.5, 5/6/247.5, 6/7/292.5, 7/0/337.5} {
    \node[covered, label={\ang:{$v_{\i\j}$}}] at (v\i\j) {};
  }
  \foreach \i/\ang in {0/0, 1/45, 2/90, 3/135, 4/180, 5/225, 6/270, 7/315} {
    \node[covered, label={\ang:{$v_\i$}}] at (v\i) {};
  }

  \node[covered_inner, label={[label distance=1pt]above:{$a_1$}}] at (a1) {};
  \node[covered_inner, label={[label distance=1pt]right:{$a_2$}}] at (a2) {};
  \node[uncovered_inner, label={[label distance=1pt]above:{$b_1$}}] at (b1) {};
  \node[uncovered_inner, label={[label distance=1pt]above:{$b_2$}}] at (b2) {};
\end{tikzpicture}

\vspace{1.5em}

\begin{tikzpicture}[scale=1.6,
  every label/.style={font=\footnotesize},
  context/.style={thick, gray!40, line cap=round},
  cover/.style={ultra thick, black, line cap=round},
  covered/.style={circle, draw=black, fill=black, minimum size=5pt, inner sep=0pt},
  covered_inner/.style={circle, draw=black, fill=black, minimum size=6pt, inner sep=0pt},
  uncovered_inner/.style={circle, draw=black, fill=white, minimum size=6pt, inner sep=0pt},
  targetA/.style={line width=3.8pt,blue,line cap=round},
  targetB/.style={line width=3.8pt,red,dash pattern=on 8pt off 10pt,line cap=round}]

  \node[font=\normalsize] at (-2.3, 2.3) {(b) $\mathcal G_-$};

  \def\Rmid{2}
  \def\Rcor{2.1648}
  \def\Rinn{0.95}

  \foreach \i/\j/\ang in {0/1/22.5, 1/2/67.5, 2/3/112.5, 3/4/157.5, 4/5/202.5, 5/6/247.5, 6/7/292.5, 7/0/337.5} {
    \coordinate (v\i\j) at (\ang:\Rcor);
  }
  \foreach \i/\ang in {0/0, 1/45, 2/90, 3/135, 4/180, 5/225, 6/270, 7/315} {
    \coordinate (v\i) at (\ang:\Rmid);
  }
  \coordinate (a1) at (180:\Rinn);
  \coordinate (a2) at (90:\Rinn);
  \coordinate (b1) at (135:\Rinn);
  \coordinate (b2) at (45:\Rinn);

  \draw[cover] (v1) -- (v5);
  \draw[cover] (v3) -- (v7);
  \draw[cover] (v70) -- (v01);
  \draw[cover] (v12) -- (v23);
  \draw[cover] (v34) -- (v45);
  \draw[cover] (v56) -- (v67);

  \draw[context] (v0) -- (v4);
  \draw[context] (v2) -- (v6);
  \draw[context] (v01) -- (v12);
  \draw[context] (v23) -- (v34);
  \draw[context] (v45) -- (v56);
  \draw[context] (v67) -- (v70);

  \draw[targetA] (a1) -- (a2);
  \draw[targetB] (b1) -- (b2);

  \foreach \i/\j/\ang in {0/1/22.5, 1/2/67.5, 2/3/112.5, 3/4/157.5, 4/5/202.5, 5/6/247.5, 6/7/292.5, 7/0/337.5} {
    \node[covered, label={\ang:{$v_{\i\j}$}}] at (v\i\j) {};
  }
  \foreach \i/\ang in {0/0, 1/45, 2/90, 3/135, 4/180, 5/225, 6/270, 7/315} {
    \node[covered, label={\ang:{$v_\i$}}] at (v\i) {};
  }

  \node[uncovered_inner, label={[label distance=1pt]above:{$a_1$}}] at (a1) {};
  \node[uncovered_inner, label={[label distance=1pt]right:{$a_2$}}] at (a2) {};
  \node[covered_inner, label={[label distance=1pt]above:{$b_1$}}] at (b1) {};
  \node[covered_inner, label={[label distance=1pt]above:{$b_2$}}] at (b2) {};
\end{tikzpicture}
\caption{The two context coverings used in the incidence-vector certificate.
Bold black lines are the selected contexts and thin gray lines are the
unselected contexts.  Panel (a) shows $\mathcal G_+$ and panel (b) shows
$\mathcal G_-$.  In both panels, the very thick blue solid line joins
$A=\{a_1,a_2\}$ and the very thick red dashed line joins $B=\{b_1,b_2\}$.  Filled
inner vertices are covered by the respective covering.}
\label{fig:octagon}
\end{figure*}
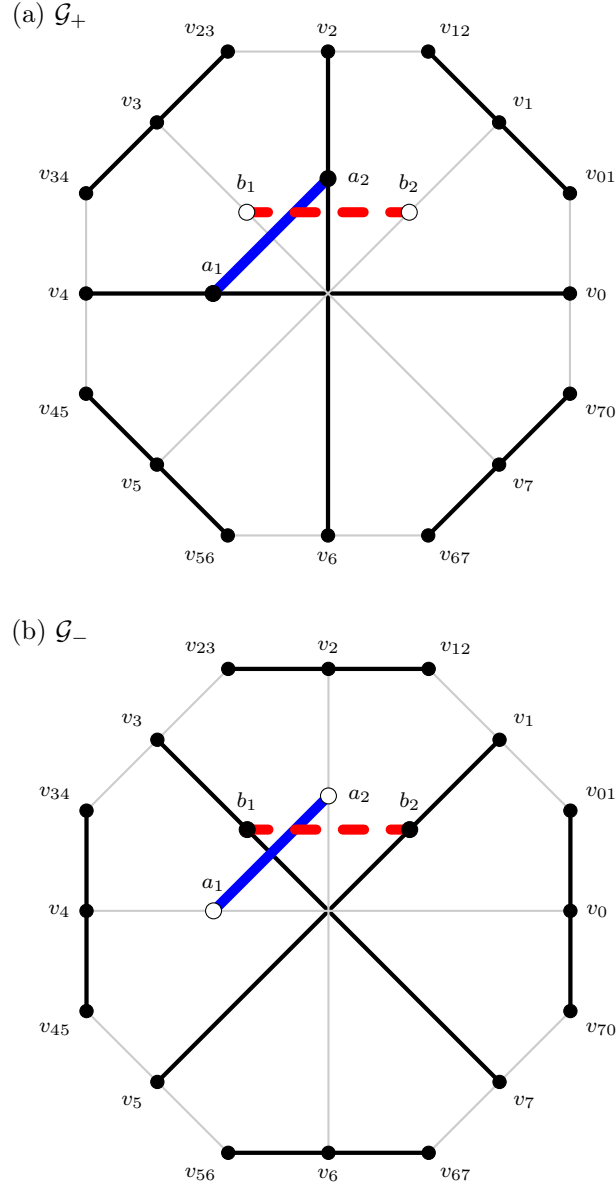

Subtraction gives the balanced certificate
\begin{equation}
 \sum_{G\in\mathcal G_+}\chi_G-
 \sum_{G\in\mathcal G_-}\chi_G
 =\chi_{a_1}+\chi_{a_2}-\chi_{b_1}-\chi_{b_2}.
 \label{eq:octagon-certificate}
\end{equation}
Thus every admissible weight \ensuremath{p}
satisfies
\begin{equation}
 p(a_1)+p(a_2)=p(b_1)+p(b_2).
 \label{eq:octagon-equality}
\end{equation}

The same equality also has a direct two-cover proof.  Let
\begin{equation*}
 S_{\rm out}(p)=\sum_{i=0}^{7}p(v_i)+
 \sum_{i=0}^{7}p(v_{i,i+1}),
\end{equation*}
where the corner indices are read modulo eight.  Adding the six
normalization equations belonging to the first cover gives
\begin{equation}
 6=S_{\rm out}(p)+p(a_1)+p(a_2),
 \label{eq:octagon-cover-plus}
\end{equation}
whereas adding those belonging to the second cover gives
\begin{equation}
 6=S_{\rm out}(p)+p(b_1)+p(b_2).
 \label{eq:octagon-cover-minus}
\end{equation}
The outer-vertex contribution is identical because each cover contains every
outer vertex exactly once.  Subtracting these two equations recovers
Eq.~\eqref{eq:octagon-equality}.  Thus the same pseudocontext identity is
visible both as the incidence-vector certificate and as the difference of two
complete covers.

The forcing proof contains no coordinates and is independent of the
Hilbert-space scalar field.  Realizability
is nevertheless field-sensitive.  The hypergraph has an exact FOR in
$\C^3$, while it has no FOR in
$\R^3$; exact coordinates and the obstruction are proved in
Ref.~\cite{2026-Navara-Svozil-complex-only-3d-hypergraphs}.  It is therefore a
pseudocontext in $\C^3$, but not in $\R^3$.  This example illustrates why
forcing and geometric admissibility must be stated as separate conditions.

\section{Minimum sizes over the real and complex fields}
\label{sec:minimal}

The following threshold is relative to Definition~\ref{def:pseudocontext},
which stipulates \ensuremath{k\ge2}.

\begin{proposition}[No size-one Hilbert-space pseudocontext]
\label{prop:no-k1}
Even if Definition~\ref{def:pseudocontext} were provisionally
extended to \ensuremath{k=1}, no triple satisfying its conditions could admit a
FOR in a real or complex Hilbert space.
\end{proposition}

\begin{proof}
Suppose such a triple had \ensuremath{A=\{a\}},
\ensuremath{B=\{b\}}, a~balanced certificate \ensuremath{\lambda},
and a FOR \ensuremath{h}.  Multiplying the projector equations for the
contexts by the corresponding coefficients \ensuremath{\lambda_E},
summing, and using the certificate equations gives
\[
 \ensuremath{P_{h(a)}-P_{h(b)}
 = (\one_\E^{\mathsf T}\lambda)I_H=0.}
\]
Equality of the rank-one projectors means equality of the rays,
contradicting injectivity of \ensuremath{h} because \ensuremath{a\ne b}.
\end{proof}

\begin{proposition}[Real rigidity at size two]
\label{prop:real-k2}
There is no pseudocontext of size two in a real Hilbert space whose two target
families are internally nonorthogonal.
\end{proposition}

\begin{proof}
Suppose a FOR supplies unit representatives
$a_1,a_2,b_1,b_2$ in a real Hilbert space satisfying
\[
 P_{a_1}+P_{a_2}=P_{b_1}+P_{b_2}=:T.
\]
The range of $T$ is the common two-dimensional span of both pairs.
Internal nonorthogonality gives
\ensuremath{\langle a_1,a_2\rangle\ne0}.  Change the sign of
\ensuremath{a_2} if necessary and set
\ensuremath{c=\langle a_1,a_2\rangle>0}.  Faithfulness makes the rays
\ensuremath{[a_1]} and \ensuremath{[a_2]} distinct, so
\ensuremath{c<1}.  Hence \ensuremath{0<c<1}.
In an eigenbasis of $T$ its eigenvalues are
$1+c$ and $1-c$, which are distinct.  Write the two unit vectors of any
decomposition in this eigenbasis as
$(\cos\phi,\sin\phi)$ and $(\cos\psi,\sin\psi)$.  The diagonal and
off-diagonal entries give
\[
 \ensuremath{\cos^2\phi+\cos^2\psi=1+c},\qquad
 \ensuremath{\sin^2\phi+\sin^2\psi=1-c},
\]
\[
 \ensuremath{\sin\phi\cos\phi+\sin\psi\cos\psi=0.}
\]
Subtracting the first two equations and applying the doubled-angle
identities yields
\[
 \ensuremath{\cos(2\phi)+\cos(2\psi)=2c},\qquad
 \ensuremath{\sin(2\phi)+\sin(2\psi)=0.}
\]
Equivalently,
\[
 \ensuremath{e^{2i\phi}+e^{2i\psi}=2c.}
\]
Two points on the unit circle with this sum are
the conjugate pair $c\pm i\sqrt{1-c^2}$.  The two rays are therefore uniquely
determined, up to interchange, by $T$.  Hence
$\{[b_1],[b_2]\}=\{[a_1],[a_2]\}$, contradicting disjointness.
Thus no allowed real size-two pair exists.
\end{proof}

Complex phases remove the real uniqueness used in the proof.  The octagon of
Section~\ref{sec:octagon} shows that this phase freedom can be enforced by a
finite context hypergraph rather than merely arranged in selected
coordinates.

\begin{theorem}[Sharp field-dependent threshold]
\label{thm:threshold}
Under Definition~\ref{def:pseudocontext}, the minimum permitted
target size \ensuremath{k=2} is attained by a pseudocontext admitting a FOR
over \ensuremath{\C}.  The minimum size over \ensuremath{\R} is
\ensuremath{3}.
\end{theorem}

\begin{proof}
The lower bound \ensuremath{k\ge2} over \ensuremath{\C}
is part of Definition~\ref{def:pseudocontext}, and the octagon
with a FOR from Section~\ref{sec:octagon} attains it over
\ensuremath{\C^3}.
Proposition~\ref{prop:real-k2} excludes size two over \ensuremath{\R},
while the construction of Section~\ref{sec:mayet} attains size three in
\ensuremath{\R^3}.
\end{proof}

\section{What the definition excludes}
\label{sec:excludes}

The stronger definition changes the classification of several familiar
operator identities.  If two unit-norm frames decompose the same positive
operator, Eq.~\eqref{eq:coordinate-identity} holds and all Born states give
the same sum.  Tight frames provide many such examples: a frame
$(a_i)_{i=1}^k$ with
\[
 \sum_i P_{a_i}=\frac{k}{n}I
\]
and a generic rotation of it yields distinct decompositions of the same
operator.  Likewise, in $\C^2$ relative phases give distinct two-ray
decompositions of one nonscalar positive operator.

These identities are \emph{representation-dependent}.  They
become pseudocontexts only if a finite auxiliary context hypergraph supplies
a certificate of the form~\eqref{eq:certificate}.  The bare rays of a generic
rotated tight-frame pair may have no contexts at all, in which case its
incidence matrix cannot force a nonzero target difference.  An operator
identity is thus a necessary consequence of a pseudocontext realized by a
FOR, by Corollary~\ref{cor:coordinates}, but it is not a sufficient
definition.

This distinction also clarifies the role of the algebraic formalization of
rank-one sums.  Coordinate calculations, including machine-checked
calculations, can verify Eq.~\eqref{eq:coordinate-identity}, trace identities,
ray distinctness, and nonorthogonality for a proposed realization.  The
additional pseudocontext claim is certified separately and exactly by the
finite incidence equation~\eqref{eq:certificate}.

Ordinary contexts are normalized independently of both coordinatization and
probability model.  Pseudocontexts inherit that same twofold independence
through a balanced finite context certificate.  A bare projector-sum identity,
by contrast, concerns Born-rule states for the chosen rank-one decomposition
and is generally coordinatization-dependent.

\section{Operational and contextuality remarks}
\label{sec:operational}

Members of one pseudocontext need not be jointly measurable.  The expression
$\sum_{a\in A}p(a)$ is a linear statistic obtained from separate, mutually
incompatible measurement arrangements and repeated preparations, not the
probability of one coarse-grained event.  It can exceed one.  The certificate asserts that
the same statistic is obtained from $B$ for every admissible preparation.

The equality is not by itself a contextuality witness.  Indeed, it holds in
every classical two-valued state whenever such states exist, as well as in
their convex mixtures.  Classical and quantum models may nevertheless allow
different ranges for the common sum, as in the examples of
Ref.~\cite{2023-navara-svozil}.  Any contextuality claim must be based on that
additional separation of state spaces, not on Eq.~\eqref{eq:universal-equality}
alone.  This is also distinct from generalized operational contextuality in
the sense of Spekkens~\cite{Spekkens-04}, where equivalences among arbitrary
preparations, transformations, or measurement procedures enter the
noncontextuality assumption.

The present equality is best regarded as a theory-independent conservation
law generated by a finite network of complete tests.  A probability model
may refine the network, restrict its state space, or supply a geometric
realization, but it cannot change a balanced incidence consequence while
retaining the same vertex identifications and context equations.

\section{Discussion and open problems}

Definition~\ref{def:pseudocontext} makes pseudocontexts finite, exact, and portable across
probability theories.  It also suggests several concrete problems.

First, one may classify the sparse signed vectors in
$\mathcal C_\Gamma$ and minimize their context certificates.  This is an
integer-linear problem at the combinatorial level.  Second, among the
resulting abstract pseudocontexts, one may ask which admit FORs in $\R^n$,
in $\C^n$, or in both.  The octagon shows that
these answers can differ even though the certificate is field-free.  Third,
one may compare the ranges of the common statistic over classical,
quantum, and more general admissible state spaces.  Only this third layer can
produce a statistical separation once the common equality has been fixed.

A further structural question is whether every rational certificate can be
replaced, perhaps after a controlled finite extension of the hypergraph, by a
$\{-1,0,1\}$ complementary covering certificate.  Such a result would turn
every abstract linear proof into a directly visible pair of coverings.
Minimal extensions that convert a representation-dependent projector
identity into a pseudocontext are another natural target.

\section{Conclusion}

A pseudocontext should not depend on lucky coordinates.  Under the
restrictive definition adopted here, the equality of its two probability
sums must be forced by finitely many completeness and additivity equations.
The balanced incidence certificate
$\ensuremath{M_\Gamma}^{\mathsf T}\lambda=\chi_A-\chi_B$,
$\ensuremath{\one_\E}^{\mathsf T}\lambda=0$ is necessary and sufficient for precisely this
notion of forcing.  It proves the equality simultaneously for classical,
quantum, softmax-generated, and other admissible weights, and it transfers to
an operator identity in every Hilbert-space FOR.

The $15$-vertex and $20$-vertex examples show that this stronger concept is
nonempty and retains a sharp real--complex minimum-size distinction.  The
separate questions ``Is the equality forced?'', ``Does the hypergraph admit
a FOR?'', and ``Which probability ranges does a theory
allow?'' should remain separate in subsequent applications.

\section*{Data availability} No datasets were generated or analyzed. All context lists and combinatorial certificates used in the proofs are included in this article. Exact coordinates for the octagon, together with verification that they realize precisely the stated orthogonality relations, are provided in Ref.~\cite{2026-Navara-Svozil-complex-only-3d-hypergraphs}.

\begin{acknowledgments}
During manuscript preparation, the authors used OpenAI Codex (GPT-5) to assist with literature organization, LaTeX editing, exact symbolic checks, and figure/PDF verification; all AI-assisted output was directed and critically reviewed by the authors, who independently verified the mathematics and take full responsibility for the final content, and no AI tool is an author.
This research was funded in part by the Austrian
Science Fund (FWF), Grant DOI \href{https://doi.org/10.55776/PIN5424624}{10.55776/PIN5424624}, and the Czech Science
Foundation (GA\v{C}R), Grant No.~25-20013L.  The authors declare no
conflict of interest.
\end{acknowledgments}

\bibliography{svozil}

\end{document}